\documentclass[review]{article}
\usepackage[utf8]{inputenc}
\usepackage{amsthm}
\usepackage{amsmath}
\usepackage{amsfonts}
\usepackage{amssymb}
\usepackage{graphicx}
\usepackage{color}
\usepackage{multirow}
\usepackage{mathtools}
\usepackage{xcolor}
\usepackage{a4wide}
\usepackage{lineno,hyperref}
\usepackage{authblk}

\newtheorem{theorem}{Theorem}
\newtheorem{lemma}[theorem]{Lemma}
\newtheorem{example}[theorem]{Example}
\newtheorem{remark}[theorem]{Remark}
\newtheorem{definition}[theorem]{Definition}
\title{Randomized Gathering of Asynchronous Mobile Robots}

\author[1]{Debasish Pattanayak}
\author[2]{John Augustine}
\author[1]{Partha Sarathi Mandal}
\affil[1]{Department of Mathematics, Indian Institute of Technology Guwahati, India}
\affil[2]{Department of Computer Science \& Engineering, Indian Institute of Technology Madras, India}
\date{}

\begin{document}
\maketitle
\begin{abstract}
	This paper revisits the widely researched \textit{gathering} problem for two robots in a scenario which allows randomization in the asynchronous scheduling model. The scheduler is considered to be the adversary which determines the activation schedule of the robots.  The adversary comes in two flavors, namely, oblivious and adaptive, based on the knowledge of the outcome of random bits. The robots follow \textit{wait-look-compute-move} cycle.
	In this paper, we classify the problems based on the capability of the adversary to control the parameters such as wait time, computation delay and the speed of robots and check the feasibility of gathering in terms of adversarial knowledge and capabilities. The main contributions include the possibility of gathering for an oblivious adversary with (i) zero computation delay; (ii) the sum of wait time and computation delay is more than a positive value. We complement the possibilities with an impossibility.
	We show that it is impossible for the robots to gather against an adaptive adversary with non-negative wait time and non-negative computation delay. Finally, we also extend our algorithm for multiple robots with merging.
\end{abstract}
    % Gathering, Mobile Robots, Asynchronous System, Randomization

\section{Introduction}
\subsection{Backgorund and Motivation}
Recently, it has piqued the interest of researchers to use small and simple robots to achieve tasks which may seem complicated for a single robot to perform.
These tiny robots, otherwise known as \textit{swarm robots}, have various advantages, namely, robustness and scalability.
Some common problems which has been addressed in the literature include
\textit{gathering}~\cite{AgmonP06,BhagatM17,Bouzid0T13,BramasT15,ChaudhuriM10,CieliebakP02,IzumiIKO13},
\textit{pattern formation}~\cite{0001FSY15,FlocchiniPSW99,FujinagaOKY10,SugiharaS96}, \textit{scattering}~\cite{BramasT17,IzumiKPT18} and \textit{flocking}~\cite{CanepaDIP16,Chaudhuri19}.
The problem of \textit{gathering} had been investigated in the presence of byzantine faults~\cite{AugerBCTU13,DefagoP0MPP16} and crash faults~\cite{Bhagat201650,BhagatM17,PattanayakMRM19}.

The widely accepted \textit{weak robots} model for mobile robots was introduced by Flocchini et al.~\cite{FlocchiniPSW99}.
The robots are generally considered to be \textit{homogeneous}, i.e., execute the same algorithm; \textit{oblivious}, i.e., do not have any knowledge of past computations; \textit{anonymous}, i.e., do not have any identifiers. No message exchange happens among the robots, hence, \textit{silent}. The robots are dimensionless \textit{point} robots.
Each robot operates in \textit{wait-look-compute-move} cycles.
\emph{Gathering} is a well-known problem that requires the robots to meet at a point which is not specified a priori.
Gathering for two robots is also known as \textit{Rendezvous} problem for two robots.
Likewise, the gathering problem can be considered to be a point formation problem.
The solution to the Rendezvous problem depends on the level of synchrony of the scheduler. There are three basic types of schedulers considered in the literature, namely, fully-synchronous (\emph{FSYNC}), semi-synchronous (\emph{SSYNC}) and asynchronous (\emph{ASYNC}).
In \textit{SSYNC}, the global time is divided into discrete rounds and a subset of robots are activated in each round by the scheduler.
\textit{FSYNC} can be considered as a special case of \textit{SSYNC}, where all the robots are activated in each round. In \textit{SSYNC}, the look states of all activated robots are aligned.

\subsection{Related Works}
The gathering problem for two robots is trivial in \emph{FSYNC} model.
It has been proved that gathering two robots is impossible with a deterministic algorithm in the \emph{SSYNC} model~\cite{SuzukiY99}.
Gathering is solved in the \emph{SSYNC} model with randomization by Izumi et al.~\cite{IzumiIKO13}, which takes a constant number of rounds in expectation with a local-weak multiplicity detection and scattered initial configuration -- all robots occupy distinct initial positions. To achieve the constant number of rounds, they assume that the scheduler is bounded, i.e., all robots activate at least once for some constant number of activations of any robot.

Lights were introduced by Das et al.~\cite{0001FPSY16} as external persistent memories to expand capabilities of robots.
In robots with lights model, Rendezvous has been solved with two color lights, which is optimal~\cite{HeribanDT18} in the \textit{ASYNC} model with non-rigid motion.
Okumura et al.~\cite{OkumuraWD18} further categorized the scheduler into multiple types including rigid and non-rigid motion of robots.
The fault-tolerant gathering of robots has been solved in the presence of one fault~\cite{AgmonP06} to multiple faults~\cite{BhagatM17,PattanayakMRM19}. Pattanayak et al.~\cite{PattanayakMRM19} also introduced the \textit{ASYNC}$_{IC}$ model -- also known as $LC$-atomic \textit{ASYNC} by Okumura et al.~\cite{OkumuraWD18}.
Randomization has also been used in~\cite{BramasT16,CanepaDIP16} to elect a leader and in~\cite{IzumiKPT18} for choosing a direction of movement in scattering.
Characterization of a continuous adversary in a discrete scenario has been addressed by Bampas et al.~\cite{BampasBCILPT19} for the Rendezvous problem of two mobile agents in graphs.

In this paper, we approach the problem of gathering robots in multiple variations of the \textit{ASYNC} model. This paper explores two different possibilities for the notion of an adversary for the randomized gathering of robots.

\subsection{Our Contribution}
We approach the problem to determine the tradeoff between randomization and asynchrony. We define the particulars of adversarial control and characterize them according to the access and knowledge of parameters such as wait time, computation delay, the speed of robots, and the outcome of random bits. We use the terms possible, improbable, and impossible to denote the results in various models. If the probability is non-zero, we say it is possible. Here improbable means that the probability of success is zero, but the set of successful outcomes is non-empty. For example, the probability of choosing 1/2 in the interval $[0,1]$ is zero, but it is possible to choose 1/2. If the set of successful outcomes is empty, then we say that it is impossible. For example, choosing 2 in the interval $[0,1]$ is impossible.

\begin{itemize}
    \item It is possible to gather two robots with non-zero wait time and zero computation delay for an oblivious adversary, where the algorithm knows the speed ratio of two robots.
    \item It is improbable for two robots to gather if the algorithm does not know the speed ratio given zero wait time and zero computation delay. 
    \item It is possible for two robots to gather if one robot has zero wait time and zero computation delay.
    \item It is also possible to gather two robots with non-negative wait time and non-negative computation delay if the sum of wait time and computation delay in any cycle is always greater than a positive value, $\tau$. The estimated number of total looks by both of the robots for gathering is $18(\log_2(\delta/\tau) + 1)$ where $\delta$ is the initial distance between them.
    \item For the adaptive adversary, we show that it is impossible for two robots to gather for non-zero wait time and computation delay.
    \item Finally, we also extend the algorithm for gathering multiple robots with merging.
\end{itemize}
We summarize the findings in Table~\ref{tab:summary}.

\begin{table}[h]\centering
	\caption{Summary of results for gathering two robots}\label{tab:summary}
	\begin{tabular}{|c|c|c|c|c|}\hline
		\textbf{Result}                     & \textbf{\shortstack{Wait                                                     \\Time} ($\mathcal{W}$)} & \textbf{\shortstack{Computation\\ Delay}($\mathcal{C}$)} & \textbf{\shortstack{Speed\\Ratio}\((\alpha)\)} & \textbf{Gathering}\\\hline
		\multicolumn{5}{c}{\bf Oblivious Adversary}                                                                        \\\hline
		Theorem~\ref{thm:SameSpeedAsyncic}  & $ \mathcal{W}\geq 0$                           & \(\mathcal{C}=0\)   & 1        & Possible   \\\hline
		Theorem~\ref{thm:AsyncicKnownAlpha} & $ \mathcal{W}\geq 0$                           & \(\mathcal{C}=0\)   & Known    & Possible   \\\hline
		Theorem~\ref{thm:UnknownAlpha}      & $\mathcal{W}= 0$                               & \(\mathcal{C}=0\)   & Unknown  & Improbable \\\hline
		Theorem~\ref{thm:R1zerozero}        & \shortstack{$\mathcal{W}_1=0$                                                       \\$\mathcal{W}_2\geq0$} & \shortstack{$\mathcal{C}_1=0$\\$\mathcal{C}_2\geq0$}& Known & Possible\\\hline
		Theorem~\ref{thm:minimumTau}        & \multicolumn{2}{c|}{$\mathcal{W} +\mathcal{C}> \tau$} & 1           & Possible              \\\hline
		                                    & $\mathcal{W} \geq 0$                           & $\mathcal{C}\geq 0$ & 1        & ?          \\\hline
		\multicolumn{5}{c}{\bf Adaptive Adversary}                                                                         \\\hline
		Theorem~\ref{thm:adaptiveasync}     & $\mathcal{W}\geq 0$                            & $\mathcal{C}\geq 0$ & 1        & Impossible \\\hline
	\end{tabular}
\end{table}
\subsection{Paper Organization}
The rest of the paper is organized as the following. In Section~\ref{sec:prelim}, we describe the model and notations. Then we describe the results of gathering for an oblivious adversary in Section~\ref{sec:oblivious}, for various models, like zero computation delay, heterogeneous speed of the robots and analyze random bit complexity. In Section~\ref{sec:adaptive}, we show the impossibility of gathering for an adaptive adversary.
In Section~\ref{sec:multirobot}, we explore the gathering of multiple robots with merging before concluding in Section~\ref{sec:conclusion}.
\section{Preliminaries}\label{sec:prelim}
\subsection{Model}
We assume the robots to be anonymous, oblivious, silent represented as points on the Euclidean plane. The robots follow the \emph{wait-look-compute-move} cycle.
When the robots are in \textit{wait} state, the robots are idle. In the \textit{look} state, the robots obtain a snapshot of the surrounding and in the \textit{compute} state they compute a destination based on the snapshot obtained. Finally, in the \textit{move} state, the robots move towards the destination. The robots do not have any agreement on the coordinate system.
There are two types of movement considered in the literature, namely, \emph{rigid} movement and \emph{non-rigid} movement.
In rigid movement, the robot reaches the destination in the same activation cycle.
In non-rigid movement, the robot moves at least a predefined distance towards the destination in an activation cycle.
In this paper, we consider the rigid movement of robots.
It is easy to observe that the non-rigid movement behaves as a rigid movement if the distance to travel becomes less than the predefined distance.

The robots have random bits with which they decide the destination as per the algorithm.
To characterize the problems, we consider the scheduler to be the \emph{adversary}.
The adversary knows the algorithm in advance, but may or may not have access to the random bits used by the algorithm in each cycle~\cite{MotwaniR95}.
Based on the above, the adversary is categorized into two types, namely, \emph{oblivious}: does not know the outcome of random bits and \emph{adaptive}: knows the outcome of \emph{past} random bits. Note that, the word \textit{oblivious} is used for denoting robots without memory of past computations and adversary without knowledge of the outcome of random bits.

We consider the general \textit{ASYNC} model for the scheduler, where each robot independently executes the \emph{wait-look-compute-move} cycle.
We consider various models which make the adversary control some aspect of a cycle.
When a robot \(R_1\) looks, it obtains the position of robot \(R_2\), which may or may not be moving.
If \(R_2\) is moving, then the position of \(R_2\) obtained at a particular instant of time that corresponds to the distance observed by \(R_1\).
So the look instant\(\mathcal{L}\) divides the time a robot remains stationary in a cycle into two parts where the former is wait time, \(\mathcal{W}\), and the latter is computation delay, \(\mathcal{C}\).
Since until the look happens and the robot knows the position of the other robot, it can be considered \textit{idle}, and after that, it is doing \textit{computation} or executing the algorithm. Fig.~\ref{fig:asyncDescription} shows the division of the look state of the \textit{ASYNC} into wait and compute.

\begin{figure}[h]
    \centering
    \includegraphics[width=\linewidth]{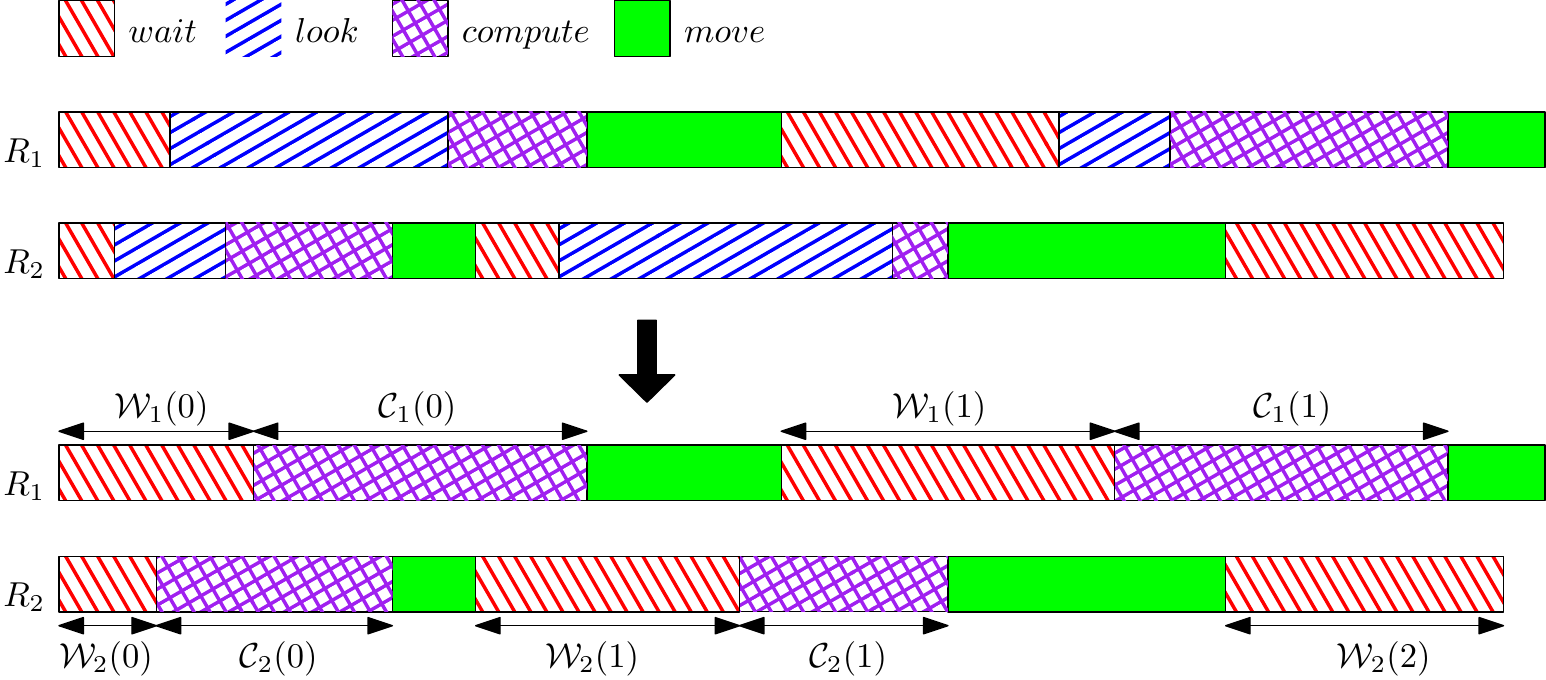}
    \caption{The top figure shows the \textit{wait-look-compute-move} steps and the bottom figure shows the division of \textit{look} state into wait time and computation delay}
    \label{fig:asyncDescription}
\end{figure}

The destination is decided by the algorithm based on the outcome of random bits and the position of the other robot.
The time required by a robot to move to its destination depends on the speed of the robot. The speeds of the robots may be the same or different. We assume that the speed of a robot remains the same throughout the execution of the algorithm.
We classify the problems based on the capabilities of the adversary to control wait time, computation delay, and the speed of robots.
When a robot looks and finds the other robot at its position, it decides that it has gathered. We assume that the robot which has decided that it has gathered stops moving thenceforth.

\subsection{Notations}
We use the following notations throughout the paper.
\begin{itemize}
    \item Time periods are denoted as $\mathcal{W}$ for wait time and $\mathcal{C}$ for computation delay.
    \item To distinguish between the two robots' wait time and computation delay, we use $\mathcal{W}_j$ and $\mathcal{C}_j$ for robot $R_j$, where $j \in \{1,2\}$.
    \item The sequence of wait time periods are denoted as $\{\mathcal{W}_j(0), \mathcal{W}_j(1), \mathcal{W}_j(2), \dots\}$ and computation delay as $\{\mathcal{C}_j(0), \mathcal{C}_j(1), \mathcal{C}_j(2), \dots\}$ corresponding to cycles $\{0, 1,2, \dots\}$ for $R_j$.
    \item We also denote the look instants as $\{\mathcal{L}_j(0), \mathcal{L}_j(1), \mathcal{L}_j(2), \dots\}$ corresponding to cycles $\{0, 1,2, \dots\}$ for $R_j$.
    \item $\delta$ is the initial distance between the robots.
    \item $\tau$ is the lower bound on the sum of wait time and computation delay for a particular robot in a cycle.
    \item $\alpha$ is the ratio of the speed of two robots.
\end{itemize}

\subsection{$\lambda$-class Algorithms}
Gathering two robots in multiple dimensions is equivalent to gathering in one dimension. We formalize the statement using the following results.

\begin{lemma}\label{lem:2dto1d}
    There exists an algorithm $\phi$ which gathers two robots in $\mathbb{R}^d$, if and only if there exists an algorithm $\phi'$ which gathers two robots in one dimension, i.e., the line joining two robots.
\end{lemma}
\begin{proof}
    We define a coordinate system for $\mathbb{R}^d$ with mid point of line joining $p_1$ and $p_2$ as the origin and $p_1$ as the positive $x$-axis.
    Suppose the destination decided by the algorithm $\phi$ in $\mathbb{R}^d$ is some point $p$.
    Take the projection of $p$ on the line joining $p_1$ and $p_2$ as $p'$.
    Moving to the position $p'$ in one dimension is the same as moving to $p$ in $\mathbb{R}^d$.
    Now, we can always find a corresponding position on the line for each destination position provided by the algorithm $\phi$.
    If the sequence of destinations gathers the two robots in $\mathbb{R}^d$, then the corresponding projections will also gather the two robots in one dimension. 
    Conversely, if an algorithm $\phi'$ gathers two robots in one dimension, i.e., the line joining them, then $\phi'$ is sufficient to gather two robots in $\mathbb{R}^d$.
\end{proof}

From Lemma~\ref{lem:2dto1d}, it is sufficient that the robots only move on the line joining two robots.
All possible algorithms for gathering two robots in one dimension can be classified based on the parameter $\lambda \in \mathbb{R}$.
In the \textit{move} state a robot can move $\lambda \delta$ distance towards the other robot, where $\delta$ is the distance between two robots.
For example, if $\lambda = 0$, a robot stays in its position, and if $\lambda = 1$, it moves to the position of the other robot.

Let us first describe the impossibility of gathering of two robots in \textit{SSYNC} model~\cite{SuzukiY99}. The capability that the adversary holds in case of \textit{SSYNC} model is whether to activate robots in a particular round or not.
\begin{example}
    If the deterministic algorithm has $\lambda = 1$, then the adversary would activate both robots in the same round.

    \begin{figure}[h!t]\centering
        \includegraphics[width=\linewidth]{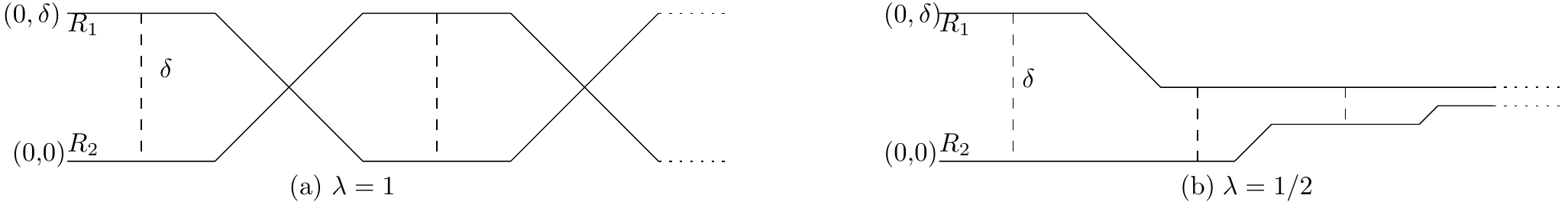}
        \caption{Execution of deterministic algorithms with $\lambda = 1/2$ and $1$ in \textit{SSYNC} model}
        \label{fig:ssyncgathering}
    \end{figure}
    If the deterministic algorithm has $\lambda = 1/2$, then the adversary would activate only one of them in a particular round. This would certainly result in reduction of the distance between the robots, but the distance would never be zero.
    Fig.~\ref{fig:ssyncgathering} shows a sample execution of the robots' behavior for $\lambda = 1$ and $\lambda = 1/2$.
    All the figures included in this paper follow a convention of representing time in the $x$-axis and distance in the $y$-axis.
    We assume $R_1$ starts at $(0, \delta)$ and $R_2$ start at $(0,0)$.
\end{example}

A randomized algorithm with $\lambda = 2$ and -1, has been presented by Izumi et al.~\cite{IzumiIKO13}, which gathers two robots in \textit{SSYNC} model.

\section{Oblivious Adversary}\label{sec:oblivious}
An oblivious adversary is not aware of the random bits produced by the algorithm.
So an oblivious adversary decides the wait time and computation delay regardless of the outcome of the sequence of random bits generated.
In other words, the adversary decides the sequence of wait times and computation delays for each robot before the start of the algorithm.
The adversary may know the ratio of the speed between the robots. In this section, we show the possibility of gathering with zero computation delay, and the speed ratio is known to the algorithm. If the speed ratio is controlled by the adversary -- unknown to the algorithm -- the adversary can ensure that the probability of gathering is zero even with zero wait time and zero computation delay (without any adversarial control). If the adversary controls one of the robots and the other has zero wait time and zero computation delay, then the probability of gathering is also positive, albeit small.
\subsection{Gathering in \textit{ASYNC}$_{\text{IC}}$ model}
First, we prove that two robots with the same speed can gather in the \emph{ASYNC}$_{IC}$ model~\cite{PattanayakMRM19}, i.e., the computation delay is zero ($\mathcal{C} = 0$). In this model, the robots move immediately after they look. For the proof, we gradually establish inequalities, which limit the choices of an adversary to prevent gathering. Finally, we show that for all choices of the adversary, there exists an algorithmic step which leads to the gathering of two robots.

\begin{theorem}\label{thm:SameSpeedAsyncic}
    For an oblivious adversary, it is possible to gather two robots with the same speed in the ASYNC$_{IC}$ model (asynchronous with zero computation delay), i.e., $\mathcal{W} \geq 0$ and $\mathcal{C} = 0$.
\end{theorem}
\begin{proof}
    Each robot when activated, looks at the other robot to determine the position of other robot and chooses a $\lambda$ value based on the following three options (with equal probability and independent of other choices): (i) $\lambda =1$, (ii) $\lambda = 1/2$, or (iii) a real value chosen uniformly at random from $(0,1)$. We now show a sequence of interactions between the algorithm and the adversary that -- with positive probability -- will guarantee that the robots will gather.

    We first introduce some notations.
    Since the computation delay is zero, let the wait time periods provided by the adversary be $\{\mathcal{W}_1(0), \mathcal{W}_1(1), \ldots\}$ and $\{\mathcal{W}_2(0), \mathcal{W}_2(1), \ldots\}$.
    We define the gap, $\gamma$ as the time difference between the look of $R_1$ and $R_2$.
    Since both robots have a speed of one unit distance per unit time, in $\gamma$ amount of time, a robot can cover $\gamma$ distance.
    If the gap $\gamma$ between look of two robots is more than the distance $\delta$, between them, and $R_2$ is idle when $R_1$ looks, then $R_1$ can reach $R_2$ before $R_2$ looks as shown in Fig.~\ref{fig:asyncicgamma}.
    In that case, to prevent gathering adversary has to choose the gap such that $\lvert\gamma\rvert < \delta$.
    If $\gamma = 0$, then both robots look at the same time and gather at mid point provided they choose $\lambda = 1/2$.
    % So we have the following.
    % \begin{equation}\label{eq:gammarange}
    %     \gamma \in (-\delta, \delta)\setminus\{0\}
    % \end{equation}
    Without loss of generality, assume that $\mathcal{W}_1(0) > \mathcal{W}_2(0)$. Denote, $\gamma_0 = \mathcal{W}_1(0)-\mathcal{W}_2(0)$ as shown in Fig.~\ref{fig:lambda1by2}.
    To avoid gathering, the adversary should choose $\mathcal{W}_1(0)$ and $\mathcal{W}_2(0)$ such that
    \begin{equation}\label{eq:gamma0}
        -\delta< \gamma_0 < \delta\text{ and }\gamma_0\neq 0
    \end{equation}
    If the adversary manages to enforce Inequality~(\ref{eq:gamma0}), we consider (as a first algorithmic step) the event that the algorithm chooses the following specific $\lambda$ values (which are anyway chosen with probability 1/9).

    \paragraph*{Step 1: $(\lambda_1 = \lambda_2 = 1)$} Though the robots start with a gap $\gamma_0$, both robots finish moving at the same time because they have the same speed.
    After both robots finish their movement, the distance between them is $\delta-\gamma_0$. The next gap between looks is defined as $\gamma_1 = \mathcal{W}_1(1) - \mathcal{W}_2(1)$ because they had finished at the same time. Further, we have the inequalities for the second look of robots if they do not gather as the following.
    \begin{equation}\label{eq:gamma1}
        -\delta + \gamma_0 < \gamma_1 < \delta-\gamma_0\text{ and }
        \gamma_1 \neq 0.
    \end{equation}
    Recall that if $\gamma_1 = 0$, then they will gather with constant probability when both choose $\lambda = 1/2$. So for the rest of the argument, we will condition on $\gamma_1 \ne 0$. Under this condition, next algorithmic step we consider (which again occurs with constant probability) is the following choice of $\lambda$ values.
    \begin{figure}[h]
        \begin{minipage}{0.45\linewidth}\centering
            \includegraphics[width=0.9\linewidth]{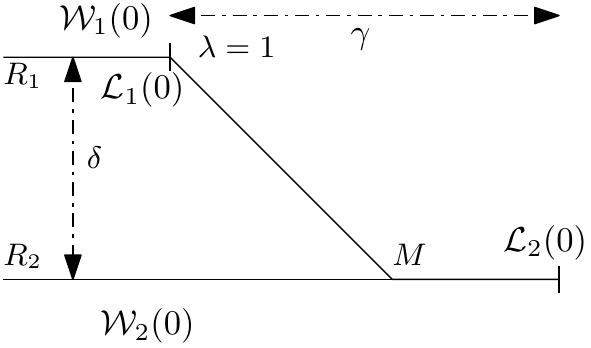}
            \caption{Gathering of two robots when the difference between looks, $\gamma$, is greater than the current distance between the robots, $\delta$.}\label{fig:asyncicgamma}
        \end{minipage}\hfill
        \begin{minipage}{0.45\linewidth}\centering
            \includegraphics[width=0.9\linewidth]{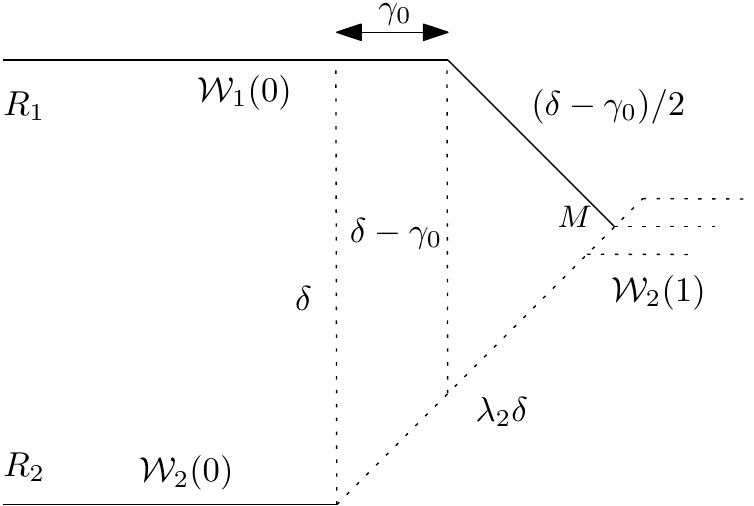}
            \caption{Gathering of two robots for $\lambda_1 = 1/2$ and $\lambda_2 \in U(0,1)$.}\label{fig:lambda1by2}.
        \end{minipage}
    \end{figure}
    \paragraph*{Step 2: $(\lambda_1 = 1/2, \lambda_2 \in U(0,1))$} Observe Fig.~\ref{fig:lambda1by2}. If $\lambda_2 = (\delta+\gamma_0)/2\delta$, then $R_1$ and $R_2$ meet at $M$. Now we estimate the probability that $R_2$ would choose $\lambda_2$ such that they will gather in the next cycle. For the robots to gather in the next cycle, whichever looks first has to find the other robot in idle state.
    \begin{description}
        \item[Case 1:] Now, the algorithm can choose $\lambda_2$ value such that it creates the situation in Fig.~\ref{fig:asyncicgamma}.
              Suppose, $R_1$ looks first in the next cycle, then it has to find $R_2$ in the idle state -- second look of $R_1$ (at $\mathcal{W}_1(0) + (\delta - \gamma_0)/2 + \mathcal{W}_1(1))$ occurs after $R_2$ finishes moving (at $\mathcal{W}_2(0) +\lambda_2\delta$) and before it looks for second time (at $\mathcal{W}_2(0) +\lambda_2\delta + \mathcal{W}_2(1)$). So,
              \begin{equation}\label{eq:lambda2range}
                  \lambda_2\delta < (\delta + \gamma_0)/2 + \mathcal{W}_1(1) < \lambda_2\delta + \mathcal{W}_2(1)
              \end{equation}
              Also, the distance between them at that point of time ($\delta - \lambda_2\delta - (\delta-\gamma_0)/2$) should be less than the gap between looks ($\lambda_2\delta + \mathcal{W}_2(1) - (\delta + \gamma_0)/2 - \mathcal{W}_1(1)$).
              \begin{equation}\label{eq:gamma1limit}
                  \lvert \lambda_2\delta - (\delta+\gamma_0)/2   \rvert < \lambda_2\delta - (\delta+\gamma_0)/2 - \gamma_1
              \end{equation}
              \begin{itemize}
                  \item
                        If $\lambda_2\delta > (\delta+\gamma_0)/2$, then $\gamma_1 < 0$ from Inequality~(\ref{eq:gamma1limit}).
                        The range of $\lambda_2$ for which Inequality~(\ref{eq:lambda2range}) and~(\ref{eq:gamma1limit}) are satisfied can be given by
                        \[
                            \lambda_2\in \Big(\max\big((\delta+\gamma_0)/2\delta, (\delta+\gamma_0+2\gamma_1)/2\delta\big), (\delta+\gamma_0+2\mathcal{W}_1(1))/2\delta\Big)
                        \]
                        Since $\gamma_1 < 0$, the probability of choosing such a $\lambda_2$ is $\mathcal{W}_1(1)/\delta$ -- the length of the smaller interval since it is chosen uniformly at random from (0,1).

                  \item  If $\lambda_2\delta < (\delta+\gamma_0)/2$, then $\lambda_2 > (\delta+\gamma_0+\gamma_1)/2\delta$ from Inequality~(\ref{eq:gamma1limit}) and $\gamma_1 < 0$ from Inequality~(\ref{eq:lambda2range}). So, we have
                        \begin{equation*}
                            \lambda_2 \in \Big(\max\big((\delta+\gamma_0+2\gamma_1)/2\delta,(\delta+\gamma_0+\gamma_1)/2\delta\big), (\delta+\gamma_0)/2\delta\Big)
                        \end{equation*}
                        The probability of choosing such a $\lambda_2$ is $-\gamma_1/2\delta$, since $\gamma_1 < 0$.
              \end{itemize}

        \item[Case 2:] Conversely, if $R_2$ looks first in the next step, then it has to find $R_1$ in the idle state. So,
              \begin{equation}\label{eq:lambda2rangeR2}
                  (\delta + \gamma_0)/2 < \lambda_2\delta + \mathcal{W}_2(1) <  (\delta + \gamma_0)/2 + \mathcal{W}_1(1)
              \end{equation}
              Similarly, the distance between them at that point of time should be less than the gap between looks.
              \begin{equation}\label{eq:gamma1limitR2}
                  \lvert \lambda_2\delta - (\delta+\gamma_0)/2   \rvert < (\delta+\gamma_0)/2 + \gamma_1 -\lambda_2\delta
              \end{equation}
              \begin{itemize}
                  \item If $\lambda_2\delta < (\delta+\gamma_0)/2$, then $\gamma_1 > 0$.
                        The range of $\lambda_2$ for which Inequality~(\ref{eq:lambda2rangeR2}) and~(\ref{eq:gamma1limitR2}) are satisfied can be given by
                        \[
                            \lambda_2\in \Big((\delta+\gamma_0-2\mathcal{W}_2(1))/2\delta, (\delta+\gamma_0)/2\delta\Big)
                        \]
                        The probability of choosing such $\lambda_2$ is $\mathcal{W}_2(1)/\delta$.
                  \item If $\lambda_2\delta > (\delta+\gamma_0)/2$, then $\lambda_2 < (\delta+\gamma_0+\gamma_1)/2\delta$. So, we have
                        \begin{equation*}
                            \lambda_2 \in \Big((\delta+\gamma_0-2\mathcal{W}_2(1))/2\delta,\min((\delta+\gamma_0+2\gamma_1)/2\delta, (\delta+\gamma_0+\gamma_1)/2\delta)\Big)
                        \end{equation*}
                        The probability of choosing such a $\lambda_2$ is $\mathcal{W}_1(1)/\delta$, if $\gamma_1 < 0$ and $(\mathcal{W}_1(1)+\mathcal{W}_2(1))/2\delta$, if $\gamma_1 > 0$.
              \end{itemize}

    \end{description}
    As the robot has a probability of $1/3$ for choosing $\lambda$ values from 1, 1/2 and $U(0,1)$, we have a multiplier of $1/27$ for the three choices of $\lambda$ values made in total by the two robots.

    Recall that $\gamma_1 \ne 0$ from Inequality~(\ref{eq:gamma1}). So the probability of gathering is dependent on the values $\mathcal{W}_1(1)$ and $\mathcal{W}_2(1)$, the adversary can choose these to be zero. Now, we consider the following algorithmic step to show a constant probability of gathering if $\mathcal{W}_1(1)$ or $\mathcal{W}_2(1)$ is equal to zero.
    \paragraph*{Step 3: $(\lambda_1 = 1/2)$}
    \begin{description}
        \item[Case 1:] If $\mathcal{W}_1(1) = 0$, then $R_1$ looks again at $M$ and finds $R_2$ at $M$ and decides that it has gathered. Once $R_1$ decides it has gathered, it does not move. We can correspondingly assign $\lambda$ values to be 1 and 1 for two consecutive looks of $R_2$ for the robots to gather. The probability of such an event is 1/27.
        \item[Case 2:] If $\mathcal{W}_2(1) = 0$, then $R_2$ can choose $\lambda$ value $1/2$. It corresponds to two situations, given the value of $\gamma_0$.
              \begin{itemize}
                  \item If $\gamma_0 \leq \delta/2$, then at the second look of $R_2$, $R_1$ is moving. So $R_2$ can choose $\lambda =1/2$ again and gather at $M$.
                  \item If $\gamma_0 > \delta/2$, then at the second look of $R_2$, $R_1$ is still in idle state. The gathering scenario repeats itself with $\gamma_0 \rightarrow \gamma_0 - \delta/2$ and $\delta\rightarrow\delta/2$. The adversary can choose $\mathcal{W}_2(l) = 0$ for $l \in \mathbb{N}$ to repeat this case. But, this situation can be repeated finitely many times before Inequality~(\ref{eq:t2i0}) holds since $\gamma_0 < \delta$ from Inequality~(\ref{eq:gamma0}).
                        \begin{equation}\label{eq:t2i0}
                            \delta/2 + \delta/4 + \dots + \delta/2^k > \gamma_0
                        \end{equation}
                        The associated probability is $1/3^{k+1}$.
              \end{itemize}
    \end{description}
    Regardless of the value of wait times chosen by adversary, we show that the probability of gathering is always positive.
    This concludes the proof.
\end{proof}
Next, we extend the possibility result of Theorem~\ref{thm:SameSpeedAsyncic} to accommodate different speed robots where the algorithm knows the speed of robots. Though it is not possible for the algorithm to distinguish between the faster robot and the slower robot, the algorithm can still utilize the ratio of speed to achieve gathering.
\begin{theorem}\label{thm:AsyncicKnownAlpha}
    It is possible to gather two robots if the algorithm knows the speed ratio, $\alpha$, of the robots for an oblivious adversary in \textit{ASYNC}\(_{IC}\) model.
\end{theorem}
\begin{proof}
    In this proof, we recreate the same situations as in the proof of Theorem~\ref{thm:SameSpeedAsyncic}.
    Similar to the previous proof, we assume $\mathcal{W}_1(0) > \mathcal{W}_2(0)$. Let us normalize the speed of $R_1$ to 1 and denote $\alpha$ as the speed ratio of $R_2$ to $R_1$.
    Choosing $\lambda_1 = 1/(\alpha+1)$ instead of $1/2$ leads to similar situation as Fig.~\ref{fig:lambda1by2}. Likewise, we can determine the probability of choosing $\lambda_2$ in a small interval around the meeting point. Since the robots are anonymous, the algorithm must include the reciprocal of $\alpha$ also. So the other value of $\lambda$ is $1/(1+1/\alpha) = \alpha/(\alpha+1)$. 
\end{proof}
\subsection{Gathering with Unknown $\alpha$}
Unlike Theorem~\ref{thm:AsyncicKnownAlpha}, if the adversary controls the speed of robots, i.e., the algorithm does not have any knowledge about the speed of robots, then it is improbable for the robots to gather. The analogy follows from the fact that the adversary can change the speed of the robots in each round such that it is unlikely for both robots meet at the same point at the same time. Here improbable stands for a zero probability event where the set of successful outcomes is non-empty. The following theorem and proof shows that the algorithm has to choose a specific real number in an interval,
\begin{theorem}\label{thm:UnknownAlpha}
    It is improbable for two robots to gather if the speed ratio $\alpha$ is controlled by the adversary with zero wait time and zero computation delay for an oblivious adversary.
\end{theorem}

\begin{proof}
    Consider the last look before two robots gather. For the two robots to gather, one robot would look at a time when the other robot is at its position. That can happen in the two scenarios, as shown in Fig.~\ref{fig:unknownalpha}.
    \begin{figure}[h]\centering
        \includegraphics[height=0.4\linewidth]{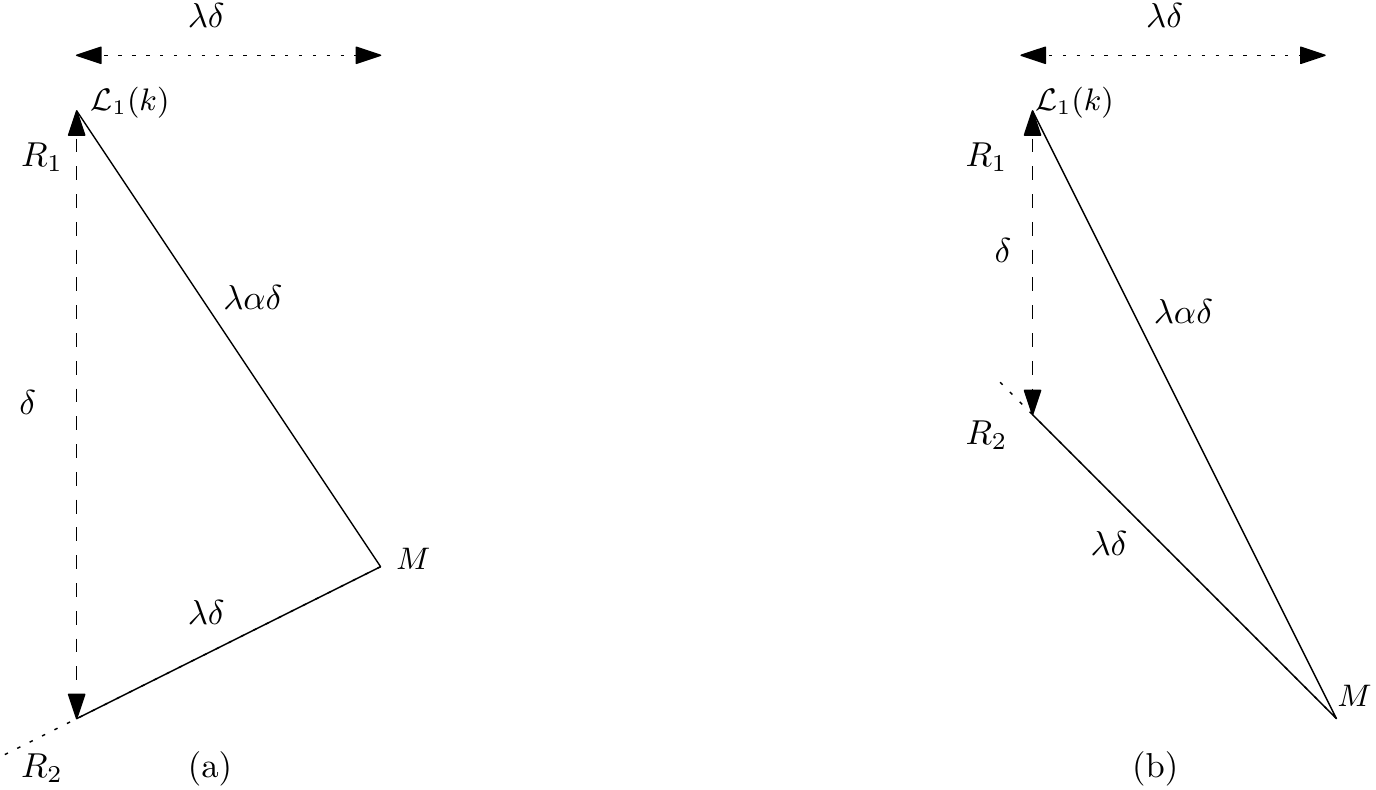}
        \caption{Improbability of gathering two robots with unknown $\alpha$}\label{fig:unknownalpha}
    \end{figure}
    Let the ratio of the speed of $R_1$ to $R_2$ is $\alpha$ with the speed of $R_2$ is normalized to 1. Here, $\alpha \in (0, \infty)$.
    For robot $R_1$ to reach the destination, it has to travel a distance of $\lambda \alpha\delta$ while the other robot travels a distance of $\lambda\delta$ in the same amount of time.
    If both the robots are traveling towards each other as shown in Fig.~\ref{fig:unknownalpha}(a), then the sum of distance covered should be $\delta$ as observed by $R_1$ at $\mathcal{L}_1(k)$,
    \[\lambda\alpha\delta + \lambda\delta = \delta \implies \lambda = 1/(1+\alpha)\]
    In the other case, if both the robots are traveling in the same direction as shown in Fig.~\ref{fig:unknownalpha}(b), then the difference of the distance should be $\delta$,
    \[\lambda\alpha\delta - \lambda\delta = \delta \implies \lambda = 1/(\alpha-1)\]
    If the algorithm chooses $\lambda \in \{1/(\alpha+1), 1/(\alpha-1)\}$, then the robots gather.
    The algorithm can choose any real number as $\lambda$ value.
    The probability of choosing $\lambda \in \{1/(\alpha+1), 1/(\alpha-1)\}$ from $(-\infty, \infty)$ is equal to zero, if the algorithm does not know the value of $\alpha$. Hence, it is improbable for the robots to gather.
\end{proof}

\subsection{Gathering of Robots when only one Robot controlled by the Adversary}
Now, we consider a situation where one robot has fixed wait time and computation delay such that the adversary cannot manipulate the robot. Following result shows the possibility of gathering.
\begin{theorem}\label{thm:R1zerozero}
    It is possible for two robots to gather if one robot has zero wait time and zero computation delay for an oblivious adversary if the algorithm knows the speed ratio $\alpha$.
\end{theorem}
\begin{figure}[h]\centering
    \includegraphics[height=0.3\linewidth]{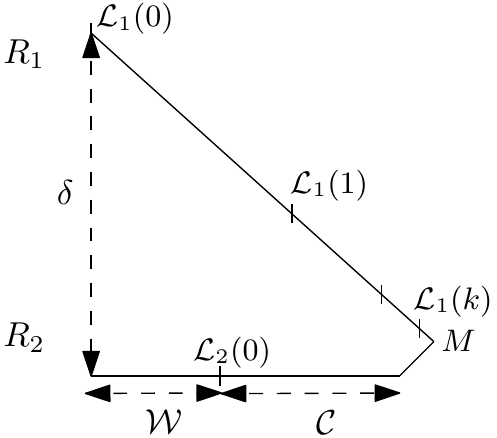}
    \caption{Gathering with $R_1$ having zero wait time and zero computation delay for $\alpha=1$}\label{fig:r1zerozero}
\end{figure}
\begin{proof}
    Consider the situation shown in Fig.~\ref{fig:r1zerozero}. Robot $R_1$ has both wait time and computation delay to be 0. We construct a scenario of gathering utilizing the situation in Theorem~\ref{thm:UnknownAlpha}.
    If $R_1$ looks at $R_2$ when $R_2$ is in the move state and chooses $\lambda=1/2$, then $R_1$ and $R_2$, both will cover half the distance before $R_1$ looks. At the next look of $R_1$, $R_1$ finds $R_2$ at its own position and decides that it has gathered. Now $R_1$ stays put. At $R_2$'s next look, it can choose $\lambda=1$ and they gather.
    For such a situation to happen, there are few constraints. $R_1$ chooses $\lambda=1/2$ continuously until it observes $R_2$ is in the move state. We have the following equation
    \[P(\delta/2 + \delta/4+\delta/8 +\cdots+\delta/2^k > \mathcal{W}+\mathcal{C}) = P(\delta-\mathcal{W}-\mathcal{C}>\delta/2^k) = 1/2^k\]
    This is true for some $k \in \mathbb{N}$. The probability of gathering is associated with the ratio $(\delta-t_i-t_c)/\delta$. Note that, for different speed robots, $R_1$ can choose $\lambda = 1/(\alpha + 1)$ instead of 1/2. Analysis would be likewise.
\end{proof}
\subsection{Gathering of two same speed robots for $\mathcal{W}+\mathcal{C} \geq \tau$}
As per the proof of Theorem~\ref{thm:SameSpeedAsyncic}, observe that the adversary can make the wait times infinitesimally smaller to have a very small probability albeit non-zero.
In this section, we restrict the adversary to prevent it from making the wait time and computation delay arbitrarily small or zero. So the adversary must choose the wait time and computation delay such that the sum of the wait time and computation delay shall be greater than a predetermined constant $\tau$. The look of a robot can occur at any time instance during that period.
Consider an algorithm that chooses one of the \(\lambda\) values 1, 1/2 and 0 with probability 1/3 each (and independent of all other choices).  First, we present some definitions and lemmas before going to the final result. 

% \begin{figure}[h]
% 	\centering
% 	\includegraphics[height=0.3\linewidth]{figtau.pdf}\caption{Distance reduction by $\tau$ when $R_1$ chooses $\lambda =1$ and $R_2$ chooses $\lambda =0$}\label{fig:tau}
% \end{figure}
% \begin{proof}
% 	The initial distance between two robots is $\delta$.
% 	Consider $\lambda$ values to be 1 and 0.
% 	For each situation where $R_1$ and $R_2$ choose different values of $\lambda$, the distance between them decreases by $\tau$ as shown in Fig.~\ref{fig:tau}. Since the $\lambda$ values 0 and 1 create a situation where one robot is not moving and the other robot is moving, we can always have the robot which finished the computation delay first as the robot which moves (i.e., $\lambda = 1$).
% 	Given the other robot has chosen $\lambda = 0$ and its computation delay ends later, regardless of the choice made by it in the next round, it will start moving 
% 	For the smallest $k \in \mathbb{N}$ such that $k\tau > \delta$, once the situation repeats \(k\) times then gathering is complete.
% \end{proof}

\begin{definition}[Maximum Distance]
    Distance $\delta$ between two robots is the \textit{maximum distance} at time $t$, if the distance at any point of time $t' > t$ does not exceed $\delta$.
\end{definition}
Note that, the instantaneous distance between the robots can be zero when the robots are crossing each other, but that does not represent the maximum distance between them.
% \begin{definition}
%     If the maximum distance between the robots decreases by at least half of the maximum distance that is termed as a \textit{phase}.
% \end{definition}

We define a sequence of attempts starting from the initiation of the algorithm. Consider the first cycle of each robot. Let $R_1$ be the robot which has moved later in the first cycle at time $t$. Then we take the latest choice of $\lambda$ value of $R_1$ and the choice of $\lambda$ value of $R_2$ immediately prior to $t$. These two looks together make an attempt. Notice that, $R_2$ can have multiple looks before $R_1$ starts moving. We only consider the latest look of $R_2$ as part of the attempt. The subsequent attempt would consider the next cycles of both robots.

\begin{definition}[Successful attempt]
    An attempt is \textit{successful} if the maximum distance between the robots is at most half of the maximum distance before the choices of $\lambda$ values done by the robots.
\end{definition}

\begin{lemma}\label{lem:phase}
    An attempt is successful with a probability at least 2/9.
\end{lemma}
\begin{proof}
    Let $\delta$ be the maximum distance between the two robots with $R_1$ at $\delta$ and $R_2$ at 0.
    Without loss of generality assume that $R_2$ is the robot which moves later.
    An attempt is successful if $R_2$ chooses $\lambda = 1/2$.
    Because, $R_1$ moves first, it can be at any position  $x\in [0, \delta]$ when $R_2$ looks. $R_1$ can choose $\lambda = 1$ or 1/2 to move.
    As $R_2$ moves to the midpoint, it will go to $x/2 \in [0, \delta/2]$.
    \begin{figure}[h]\centering
        \includegraphics[width=\linewidth]{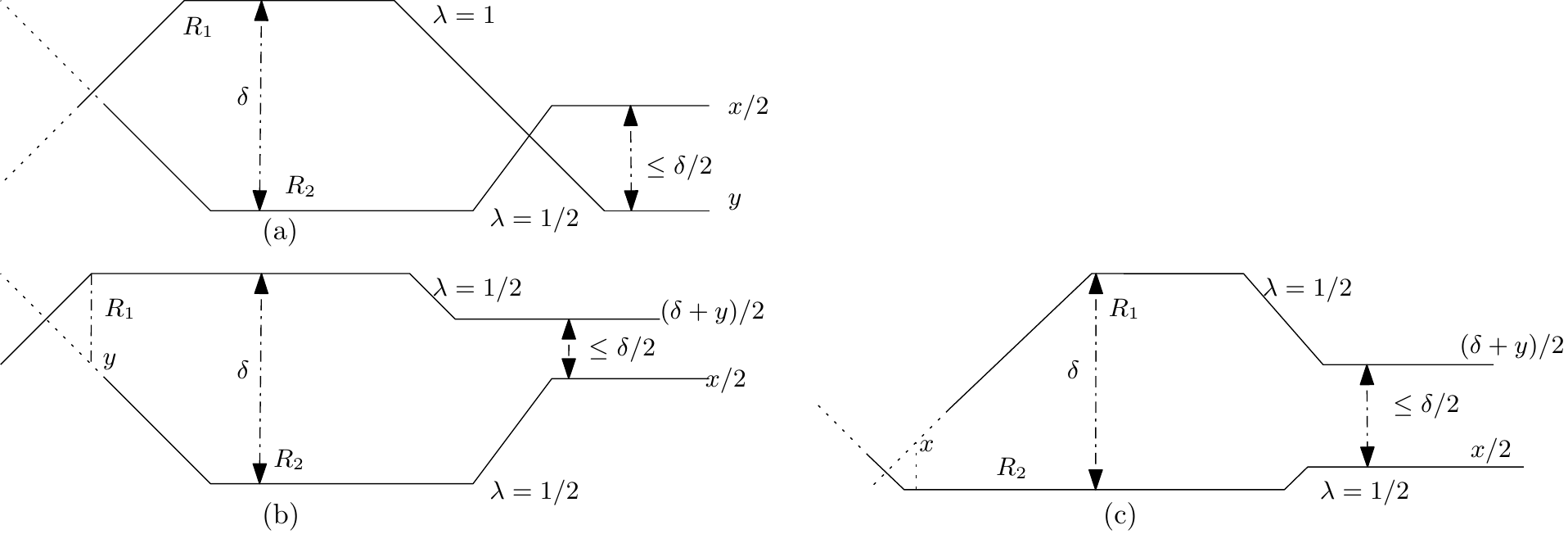}\caption{Occurrence of a phase when $R_2$ choses $\lambda = 1/2$ and moves later}\label{fig:phase}
    \end{figure}
    \begin{itemize}
        \item If $R_1$ has chosen $\lambda = 1$, then the robots would cross each other. The maximum distance between them is $x/2 \leq \delta/2$ if $R_1$ reaches 0 as shown in Fig.~\ref{fig:phase}(a).
        \item If $R_1$ has chosen $\lambda = 1/2$, then the robots may or may not cross each other. If the robots cross each other, it is similar to the previous case with $\lambda = 1$ for $R_1$. If the robots do not cross each other, there are two situations depending on which robot looks first. 
        \begin{itemize}
            \item If $R_1$ looks first, then it may find $R_2$ before $R_2$ reaches 0, i.e., before $R_2$ is at the farthest distance from $R_1$.
            Let $R_2$ be at $y$ when $R_1$ has looked. Now, $y \in [0,\delta)$. Since $R_1$ moves first, it finishes its movement first as it is moving only half the distance.
            Now, $R_2$ can observe $R_1$ at any position between $x \in[(\delta+y)/2,\delta]$. So, the next distance between them is $(\delta + y)/2 - x/2 \leq (\delta + y)/2 - (\delta+ y)/4 \leq \delta/2$ as shown in Fig.~\ref{fig:phase}(b).
            \item If $R_2$ looks first, then $y \in [0, x/2]$. The next distance between them is $(\delta + y)/2 - x/2 \leq \delta/2 - x/4 \leq \delta/2$ as shown in Fig.~\ref{fig:phase}(c).
        \end{itemize}
    \end{itemize}
    Hence, the maximum distance of $R_1$ from $R_2$ is bound by $\delta/2$.
    The probability for $R_1$ to choose $\lambda$ to be 1 or 1/2 is 2/3 and for $R_2$ to choose $\lambda = 1/2$ is 1/3. So the probability of a successful attempt is at least $2/3\times 1/3 = 2/9$.
\end{proof}

\begin{definition}[Phase]
    A \textit{phase} is a set of consecutive attempts until a successful attempt.
\end{definition}

\begin{lemma}\label{lem:looksinaphase}
    A \textit{phase} contains at most 18 looks in expectation. 
\end{lemma}

\begin{proof}
    Let $X_i$ be a random variable which denotes the number of looks in a phase. Since a phase contains successive attempts until a successful attempt, let $Y_i$ be the random variable denoting the number of attempts in a phase.
    From Lemma~\ref{lem:phase}, each attempt has a success probability of 2/9. Each attempt is independent of the previous attempts. So the number of attempts before a successful attempt follows a geometric distribution. Thus, $E[Y_j] = 1/(2/9) = 9/2$.
    Let $Z$ be the random variable denoting the number of looks in an attempt. By definition, $Z \geq 2$. An attempt will finish if the robot which started first chooses $\lambda$ value to be 1. The number of looks before the robot chooses $\lambda = 1$ also follows a geometric distribution with expectation 3. So, the expected number of looks in an attempt is 4.
    The number of looks in a phase is, $X_i = \sum_{j = 1}^Z Y_j$.
    By Wald's Equation~\cite{wald1944}, $E[X_i] = E[Y_j]E[Z] = 4\times 9/2 = 18$.
\end{proof}

\begin{theorem}\label{thm:minimumTau}
    It is possible to gather two robots if the sum of wait time and computation delay is always greater than some value $\tau$ for both the robots having the same speed.
\end{theorem}

\begin{proof}
Let $T$ be a random variable which denotes the number of phases required such that the maximum distance between two robots is less than $\tau$.
As each phase reduces the distance between robots by at least half the maximum distance in the previous phase, we can say that $T \leq \log_2(\delta/\tau)$.
Let us estimate the total number of looks required for the maximum distance to become less than $\tau$. Each look corresponds to a choice of $\lambda$ value, which also signifies the random bit complexity of the process as a whole.
Let $X_i$ be a random variable which represents the number of looks in each phase. 
We have $E[X_i] = 18$ from Lemma~\ref{lem:looksinaphase}.
Let random variable $X$ denote the total number of looks required for maximum distance to be less than $\tau$. So,
\[
    X = \sum_{i=1}^T X_i
\]
Using Wald's Equation~\cite{wald1944}, we have 
\[
    E[X] = E[X_i]E[T] = 18 E[T] \leq 18\log_2(\delta/\tau)
\]

\begin{figure}[h]\centering
    \includegraphics[height=0.2\linewidth]{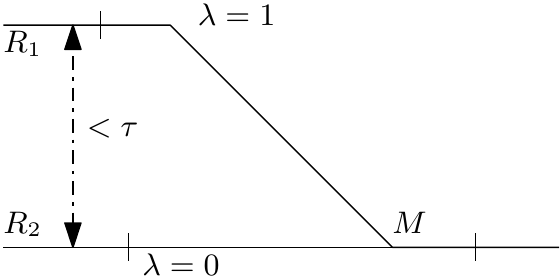}\caption{Gathering of two robots when distance between them is less than $\tau$.}\label{fig:taufinal}
\end{figure}
Next, we describe the probability of gathering after the maximum distance between robots in a phase is less than $\tau$. If the robot moving first chooses $\lambda = 1$ and the other chooses $\lambda = 0$ in the next attempt, that is sufficient for gathering as shown in Fig.~\ref{fig:taufinal}. 
The gathering happens definitely because whichever robot looks after the robot moving first has finished moving, would find the other at the same location and decide that they have gathered.
The probability of such an event is $1/9$. Hence, the expected number of attempts before the robots gather is $9$. The number of looks for such an event is 2. All other combinations of $\lambda$ values in two looks may not lead to gathering. So, the corresponding number of looks is $18$.
Expected number of looks required by both robots combined is at most $18(\log_2(\delta/\tau) + 1)$.
This concludes the proof.
\end{proof}
\begin{remark}
    Two robots can gather with two different values of $\lambda$ (e.g., 0 and 1), so the number of random bits required per cycle is 1, which is optimal.
\end{remark}

\section{Adaptive Adversary}\label{sec:adaptive}
An adaptive adversary is aware of all the outcome of random bits that appeared; consequently, it knows the destination of a robot in a cycle after the look instant.
The adversary decides the computation delay of the current cycle and wait time of the next cycle for a robot after the look instant.
\begin{figure}[h]\centering
    \includegraphics[height=0.2\linewidth]{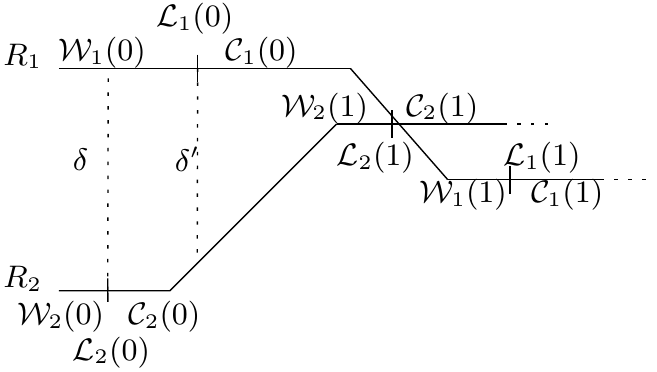}
    \caption{$R_1$ and $R_2$ look each other during movement}\label{fig:asyncAdaptive}
\end{figure}
\begin{theorem}\label{thm:adaptiveasync}
    It is impossible for two robots with the same speed to gather in \textit{ASYNC} model with an adaptive adversary.
\end{theorem}
\begin{proof}
    The adaptive adversary can create a situation such that at each look of a robot, it looks at the other robot in move state.
    Suppose, $R_1$ is looking and $R_2$ is in move state, based on $R_1$'s position $R_2$ decides to move to some position. After $R_2$ starts moving, $R_1$ looks. And this cycle continues perpetually.

    Suppose, the adversary has chosen $\mathcal{W}_1(0) > \mathcal{W}_2(0)$.
    As shown in Fig.~\ref{fig:asyncAdaptive}, $R_2$ looks first at $\mathcal{L}_2(0)$, when $R_1$ is idle. So, $R_2$ decides the value of $\lambda_2$ at $\mathcal{L}_2(0)$. Since the adversary is adaptive, it has knowledge of $\lambda_2$. Now, the adversary decides the value of $\mathcal{C}_2(0)$ and $\mathcal{W}_2(1)$ such that
    \[
        \mathcal{W}_2(0) + \mathcal{C}_2(0) < \mathcal{W}_1(0)
    \]
    Note that, adversary can choose $\mathcal{W}_2(1)$ to be any non-negative real number.
    Next, $R_1$ looks at $\mathcal{L}_1(0)$ and $R_2$ is in move state. Since, the adversary already knows the time at which $R_2$ looks next time, i.e., $\mathcal{L}_2(1)$, it can choose $\mathcal{C}_1(0)$ such that
    \begin{align*}
        \mathcal{W}_2(0) + \mathcal{C}_2(0) + \lambda_2\delta + \mathcal{W}_2(1) > \mathcal{W}_1(0) + \mathcal{C}_1(0)    \text{ and }\\
        \mathcal{W}_2(0) + \mathcal{C}_2(0) + \lambda_2\delta + \mathcal{W}_2(1) < \mathcal{W}_1(0) + \mathcal{C}_1(0) + \lambda_1\delta'
    \end{align*}
    where, $\delta'$ is the distance of $R_2$ from $R_1$ at $\mathcal{L}_2(0)$.
    Continuing in this manner, the adversary can guarantee that at each look apart from $\mathcal{L}_2(0)$ each robot finds the other in move state.

    Also, as a special case, if $\lambda$ chosen by a robot is 0, then the adversary can choose the next computation delay and wait time to be zero. Then the robot has to choose $\lambda$ value again. Hence, it is impossible for two robots to gather with an adaptive adversary.
\end{proof}

\section{Gathering multiple robots with merging}\label{sec:multirobot}
We define merging as once multiple robots are collocated, they will behave as a single entity thenceforth. With merging, we can reduce the multi-robot gathering problem to a two robot gathering problem for which the solution has been explored in the previous sections.
First, we describe Algorithm~\texttt{ReduceToLine}, which reduces the gathering problem in $\mathbb{R}^2$ to gathering on a line. Then we describe the reduction to a scenario where the robots are located at two distinct positions.

\subsection{Reduce to Line}
The robots are on a Euclidean plane. Determine the farthest pair of robots. If there are multiple pair of robots with the maximum distance, each activated robot moves $\lambda/100$ times the distance between them linearly outwards along the line joining them. Here 100 is chosen as an arbitrary constant. This movement makes one pair of robots as the unique maximum distance pair. The adversary can activate all the robots in the farthest pairs at the same time so that there can be multiple pairs of robots with the maximum distance in the next cycle. 
Specifically, if there are $k$ pairs of robots with the same distance and two different values of $\lambda$, then it would take $O(\log_2(k))$ rounds to arrive at a unique pair of robots which are the farthest distance from each other.
The line joining the maximum distance pair of robots is unique. Once the unique line is determined, each robot determines its projection on the line and moves towards it. Note that, moving towards the projection cannot make another pair of robots as the maximum distance pair, as the projected distance of a pair of robots is smaller compared to their original distance. Once the process is complete, the robots will form a line.

\subsection{With Merging}
If two robots merge and behave as a single robot, then we can deterministically achieve a two robot gathering scenario. The inner robots on the line do not move. Observe that, an intermediate robot can see other robots in two opposite directions. The outermost robots move towards the nearest inner robot, and this continues recursively. A robot can determine that it is the outermost robot if it sees other robots only in one direction. It is possible to achieve gathering directly from a situation with three distinct robot positions, say $A$, $B$, and $C$. It can happen if and only if the activation time of the robot at the middle position $B$ does not lie between the arrival time of $A$ at $B$ and $C$ at $B$. Otherwise, it becomes a two robot gathering problem.

\subsection{Discussion for gathering without merging}
If merging is not allowed, then the collocated robots behave as individual entities.
The process of forming a configuration with two distinct robot locations from a multi-robot configuration on the line is similar to the previous subsection.
The difficulty here lies in the fact that the configuration with two distinct robot locations can diverge further into a configuration with multiple robot locations. 
Once a subset of robots at some position are activated by the adversary, they start moving according to the configuration with two robot locations. 
While the robots are in movement, if other robots are activated subsequently, then the configuration becomes a configuration with multiple robot location for the robots activated later, and it will cause the robots activated latter to take decision based on a configuration with multiple robot location.

\section{Conclusion}\label{sec:conclusion}
This paper sheds light on the gathering of two robots with random bits. The adversary is characterized into two types, oblivious and adaptive. In case of an oblivious adversary, we have proved the possibility of gathering in the \textit{ASYNC$_{IC}$} model for a known speed ratio of the two robots.
With a lower bound on adversarial control, we also provide an estimate for the expected number of looks.
The approach of gathering has been generalized to include multiple robots with merging.
It remains to be investigated whether it is possible to gather in the general \textit{ASYNC} model with no restrictions on an oblivious adversary.
We also show that an adaptive adversary can always prevent the gathering of the robots in \textit{ASYNC}.

As future work, further characterization of choice of $\lambda$ values can be considered, and the analysis can be improved. One can also explore potential adversarial behaviors in the gathering of multiple robots without merging.
\bibliography{bib}

\begin{thebibliography}{10}

\bibitem{AgmonP06}
Noa Agmon and David Peleg.
\newblock Fault-tolerant gathering algorithms for autonomous mobile robots.
\newblock {\em {SIAM} J. Comput.}, 36(1):56--82, 2006.

\bibitem{AugerBCTU13}
C{\'{e}}dric Auger, Zohir Bouzid, Pierre Courtieu, S{\'{e}}bastien Tixeuil, and
  Xavier Urbain.
\newblock Certified impossibility results for byzantine-tolerant mobile robots.
\newblock In {\em Stabilization, Safety, and Security of Distributed Systems -
  15th Intl. Symposium, {SSS} 2013, Osaka, Japan, November 13-16, 2013.
  Proceedings}, pages 178--190, 2013.

\bibitem{BampasBCILPT19}
Evangelos Bampas, L{\'{e}}lia Blin, Jurek Czyzowicz, David Ilcinkas, Arnaud
  Labourel, Maria Potop{-}Butucaru, and S{\'{e}}bastien Tixeuil.
\newblock On asynchronous rendezvous in general graphs.
\newblock {\em Theor. Comput. Sci.}, 753:80--90, 2019.

\bibitem{Bhagat201650}
S.~Bhagat, S.~Gan Chaudhuri, and K.~Mukhopadhyaya.
\newblock Fault-tolerant gathering of asynchronous oblivious mobile robots
  under one-axis agreement.
\newblock {\em J. Discrete Algorithms}, 36:50 -- 62, 2016.

\bibitem{BhagatM17}
Subhash Bhagat and Krishnendu Mukhopadhyaya.
\newblock Fault-tolerant gathering of semi-synchronous robots.
\newblock In {\em Proceedings of the 18th International Conference on
  Distributed Computing and Networking, Hyderabad, India, January 5-7, 2017},
  page~6, 2017.

\bibitem{Bouzid0T13}
Zohir Bouzid, Shantanu Das, and S{\'{e}}bastien Tixeuil.
\newblock Gathering of mobile robots tolerating multiple crash faults.
\newblock In {\em {IEEE} 33rd Intl. Conference on Distributed Computing
  Systems, {ICDCS} 2013, 8-11 July, 2013, Philadelphia, Pennsylvania, {USA}},
  pages 337--346, 2013.

\bibitem{BramasT15}
Quentin Bramas and S{\'{e}}bastien Tixeuil.
\newblock Wait-free gathering without chirality.
\newblock In {\em Structural Information and Communication Complexity - 22nd
  International Colloquium, {SIROCCO} 2015, Montserrat, Spain, July 14-16,
  2015, Post-Proceedings}, pages 313--327, 2015.

\bibitem{BramasT16}
Quentin Bramas and S{\'{e}}bastien Tixeuil.
\newblock Brief announcement: Probabilistic asynchronous arbitrary pattern
  formation.
\newblock In {\em Proceedings of the 2016 {ACM} Symposium on Principles of
  Distributed Computing, {PODC} 2016, Chicago, IL, USA, July 25-28, 2016},
  pages 443--445, 2016.

\bibitem{BramasT17}
Quentin Bramas and S{\'{e}}bastien Tixeuil.
\newblock The random bit complexity of mobile robots scattering.
\newblock {\em Int. J. Found. Comput. Sci.}, 28(2):111--134, 2017.

\bibitem{CanepaDIP16}
Davide Canepa, Xavier D{\'{e}}fago, Taisuke Izumi, and Maria Potop{-}Butucaru.
\newblock Flocking with oblivious robots.
\newblock In {\em Stabilization, Safety, and Security of Distributed Systems -
  18th International Symposium, {SSS} 2016, Lyon, France, November 7-10, 2016,
  Proceedings}, pages 94--108, 2016.

\bibitem{Chaudhuri19}
Sruti~Gan Chaudhuri.
\newblock Flocking along line by autonomous oblivious mobile robots.
\newblock In {\em Proceedings of the 20th International Conference on
  Distributed Computing and Networking, {ICDCN} 2019, Bangalore, India, January
  04-07, 2019}, pages 460--464, 2019.

\bibitem{ChaudhuriM10}
Sruti~Gan Chaudhuri and Krishnendu Mukhopadhyaya.
\newblock Gathering asynchronous transparent fat robots.
\newblock In {\em Distributed Computing and Internet Technology, 6th
  International Conference, {ICDCIT} 2010, Bhubaneswar, India, February 15-17,
  2010. Proceedings}, pages 170--175, 2010.

\bibitem{CieliebakP02}
Mark Cieliebak and Giuseppe Prencipe.
\newblock Gathering autonomous mobile robots.
\newblock In {\em {SIROCCO} 9, Proceedings of the 9th Intl. Colloquium on
  Structural Information and Communication Complexity, Andros, Greece, June
  10-12, 2002}, pages 57--72, 2002.

\bibitem{0001FPSY16}
Shantanu Das, Paola Flocchini, Giuseppe Prencipe, Nicola Santoro, and Masafumi
  Yamashita.
\newblock Autonomous mobile robots with lights.
\newblock {\em Theor. Comput. Sci.}, 609:171--184, 2016.

\bibitem{0001FSY15}
Shantanu Das, Paola Flocchini, Nicola Santoro, and Masafumi Yamashita.
\newblock Forming sequences of geometric patterns with oblivious mobile robots.
\newblock {\em Distributed Computing}, 28(2):131--145, 2015.

\bibitem{DefagoP0MPP16}
Xavier D{\'{e}}fago, Maria~Gradinariu Potop{-}Butucaru, Julien Cl{\'{e}}ment,
  St{\'{e}}phane Messika, Philippe~Raipin Parv{\'{e}}dy, and Philippe~Raipin
  Parv{\'{e}}dy.
\newblock Fault and byzantine tolerant self-stabilizing mobile robots gathering
  - feasibility study -.
\newblock {\em CoRR}, abs/1602.05546, 2016.

\bibitem{FlocchiniPSW99}
Paola Flocchini, Giuseppe Prencipe, Nicola Santoro, and Peter Widmayer.
\newblock Hard tasks for weak robots: The role of common knowledge in pattern
  formation by autonomous mobile robots.
\newblock In {\em Proc. ISAAC}, pages 93--102, 1999.

\bibitem{FujinagaOKY10}
Nao Fujinaga, Hirotaka Ono, Shuji Kijima, and Masafumi Yamashita.
\newblock Pattern formation through optimum matching by oblivious {CORDA}
  robots.
\newblock In {\em Principles of Distributed Systems - 14th International
  Conference, {OPODIS} 2010, Tozeur, Tunisia, December 14-17, 2010.
  Proceedings}, pages 1--15, 2010.

\bibitem{HeribanDT18}
Adam Heriban, Xavier D{\'{e}}fago, and S{\'{e}}bastien Tixeuil.
\newblock Optimally gathering two robots.
\newblock In {\em Proceedings of the 19th International Conference on
  Distributed Computing and Networking, {ICDCN} 2018, Varanasi, India, January
  4-7, 2018}, pages 3:1--3:10, 2018.

\bibitem{IzumiIKO13}
Taisuke Izumi, Tomoko Izumi, Sayaka Kamei, and Fukuhito Ooshita.
\newblock Feasibility of polynomial-time randomized gathering for oblivious
  mobile robots.
\newblock {\em {IEEE} Trans. Parallel Distrib. Syst.}, 24(4):716--723, 2013.

\bibitem{IzumiKPT18}
Taisuke Izumi, Daichi Kaino, Maria~Gradinariu Potop{-}Butucaru, and
  S{\'{e}}bastien Tixeuil.
\newblock On time complexity for connectivity-preserving scattering of mobile
  robots.
\newblock {\em Theor. Comput. Sci.}, 738:42--52, 2018.

\bibitem{MotwaniR95}
Rajeev Motwani and Prabhakar Raghavan.
\newblock {\em Randomized Algorithms}.
\newblock Cambridge University Press, 1995.

\bibitem{OkumuraWD18}
Takashi Okumura, Koichi Wada, and Xavier D{\'{e}}fago.
\newblock Optimal rendezvous {L}-algorithms for asynchronous mobile robots with
  external-lights.
\newblock In {\em 22nd International Conference on Principles of Distributed
  Systems, {OPODIS} 2018, December 17-19, 2018, Hong Kong, China}, pages
  24:1--24:16, 2018.

\bibitem{PattanayakMRM19}
Debasish Pattanayak, Kaushik Mondal, H.~Ramesh, and Partha~Sarathi Mandal.
\newblock Gathering of mobile robots with weak multiplicity detection in
  presence of crash-faults.
\newblock {\em J. Parallel Distrib. Comput.}, 123:145--155, 2019.

\bibitem{SugiharaS96}
Kazuo Sugihara and Ichiro Suzuki.
\newblock Distributed algorithms for formation of geometric patterns with many
  mobile robots.
\newblock {\em J. Field Robotics}, 13(3):127--139, 1996.

\bibitem{SuzukiY99}
Ichiro Suzuki and Masafumi Yamashita.
\newblock Distributed anonymous mobile robots: Formation of geometric patterns.
\newblock {\em {SIAM} J. Comput.}, 28(4):1347--1363, 1999.

\bibitem{wald1944}
Abraham Wald.
\newblock On cumulative sums of random variables.
\newblock {\em Ann. Math. Statist.}, 15(3):283--296, 09 1944.

\end{thebibliography}
\end{document}